\documentclass[12pt]{article}
\usepackage{amsmath}
\usepackage{graphicx,psfrag,epsf}
\usepackage{enumerate}
\usepackage{psfrag,epsf}
\usepackage{url} % not crucial - just used below for the URL
\usepackage{amsmath,amssymb,amsfonts,amsthm,mathtools,mathrsfs}

\usepackage{booktabs}
\usepackage{bm,bbm}
\usepackage{color}
\usepackage{subfigure}
\usepackage{algorithm,algpseudocode}
\usepackage[symbol]{footmisc}
\usepackage{scalefnt}
\usepackage{authblk} % author and affiliations
\usepackage{multirow,centernot}
\usepackage[colorlinks=true, citecolor=blue, urlcolor=blue]{hyperref}
\usepackage{natbib}
\usepackage{dsfont}
\usepackage[title]{appendix}
\usepackage{newtxtext,newtxmath} % Times Roman fonts (text and math)
\usepackage{mhchem} % for \ce macro

\input cyracc.def

\usepackage[left=1in,top=1.1in,right=0.5in,bottom=1in]{geometry}

\theoremstyle{definition}

\newtheorem{theorem}{Theorem}[section]

\newtheorem{corollary}[theorem]{Corollary}

\newtheorem{proposition}[theorem]{Proposition}
\newtheorem{remark}[theorem]{Remark}
\makeatletter
\def\@seccntformat#1{\@ifundefined{#1@cntformat}%
	{\csname the#1\endcsname\quad}%      default
	{\csname #1@cntformat\endcsname}%    enable individual control
}
\makeatother

\newif\ifShowComments
\ShowCommentstrue
\def\strutdepth{\dp\strutbox}
\def\druk#1{\strut\vadjust{\kern-\strutdepth
        {\vtop to \strutdepth{%
                \baselineskip\strutdepth\vss
                        \llap{\hbox{#1}\quad}\null}}}}

\title{\bf
An unbiased estimator of a novel extended Gini index for gamma distributed populations
}

\author{
\text{Roberto Vila}$^{1}$\thanks{Corresponding author: Roberto Vila, email: {rovig161@gmail.com}
}
\,\,\,and
\text{Helton Saulo}$^{1,2}$
\\
{\small $^{1}$ Department of Statistics, University of Brasilia, Brasilia, Brazil}\\
{\small $^{2}$ Department of Economics, Federal University of Pelotas, Pelotas, Brazil}\\
}
\begin{document}
	\maketitle 	
	\begin{abstract}
In this paper, we introduce a novel flexible Gini index, referred to as the extended Gini index, which is defined through ordered differences between the $j$th and $k$th order statistics within subsamples of size $m$, for indices satisfying $1 \leqslant j \leqslant k \leqslant m$.  We derive a closed-form expression for the expectation of the corresponding estimator under the gamma distribution and prove its unbiasedness, thereby extending prior findings by \cite{Deltas2003}, \cite{Baydil2025}, and \cite{Vila2025}. A Monte Carlo simulation illustrates the estimator's finite-sample unbiasedness. A real
data set on gross domestic product (GDP) per capita is analyzed to illustrate the proposed measure.

	\end{abstract}
	\smallskip
	\noindent
	{\small {\bfseries Keywords.} {Gamma distribution, extended Gini index, extended Gini index estimator, $m$th Gini index, unbiased estimator.}}
	\\
	{\small{\bfseries Mathematics Subject Classification (2010).} {MSC 60E05 $\cdot$ MSC 62Exx $\cdot$ MSC 62Fxx.}}
%	
%%	\tableofcontents

\section{Introduction}

Measuring income inequality is a central task in economics and public policy. The classical Gini index is \citep{Gini1936} a widely used inequality measure.  Nevertheless, this index suffers from a key limitation, i.e., distinct income distributions with different concentration patterns may yield the same Gini index value. Recently, \cite{Gavilan-Ruiz2024} addressed this problem by deriving an index computed as the expected difference between the maximum and minimum values in random samples of size $m$. This novel $m$th Gini index retains the intuitive appeal of the classical Gini measure while enabling finer discrimination between distributions that share the same Gini value but differ in their internal inequality structure.

In this paper, we propose a novel generalization, called the extended Gini index, which is defined using ordered differences between the $j$th and $k$th values within subsets of size $m$, satisfying $1 \leqslant j \leqslant k \leqslant m$. This formulation allows for a more flexible treatment of the data by systematically discarding the most extreme values. We prove that the associated extended Gini index estimator is unbiased when the underlying population follows a gamma distribution. Our result extends and unifies previous findings: \cite{Deltas2003} established that the sample Gini coefficient is unbiased under an exponential distribution; \cite{Baydil2025} showed that the sample Gini coefficient remains unbiased under gamma-distributed populations; and \cite{Vila2025} proved that the $m$th Gini index estimator
%, originally introduced by \cite{Gavilan-Ruiz2024}, 
is unbiased in the gamma case.

The motivation for introducing the extended Gini index stems from the need to better capture inequality patterns in specific distributional scenarios. Consider, for example, income distributions where the majority of individuals are concentrated in lower (or higher) income levels, but a small fraction of the population belongs to the opposite extreme--very rich or very poor. In such cases, the $m$th Gini index \citep{Gavilan-Ruiz2024} may yield inflated or distorted inequality measures due to the influence of these few extreme values. Our proposed extended Gini index addresses this issue by allowing the analyst to disregard the outermost values in each subset through the choice of indices $1 \leqslant j \leqslant k \leqslant m$. This flexibility leads to a more realistic quantification of inequality, especially when the extremes are not representative of the broader population. Therefore, by focusing on central order statistics, the extended index can report significantly lower inequality levels in cases where most individuals have similar incomes, and only a few lie at the extremes.

The rest of this paper unfolds as follows. In Section \ref{sec:02}, we define the proposed extended Gini index and obtain an explicit expression under the assumption of a gamma-distributed population. In Section \ref{sec:03}, we establish the unbiasedness of the sample extended Gini index estimator and provide a closed-form expression for its expectation. In Section \ref{sec:04}, we present a Monte Carlo simulation study that illustrates and supports our theoretical results. In Section \ref{sec:05}, we apply the proposed estimator to a data set on  on gross domestic product (GDP) per capita for demonstrating the practical utility of the measure introduced here. Finally, in Section \ref{sec:06}, we offer concluding remarks.

\section{Extended Gini index}\label{sec:02}

Let $X_1, X_2, \ldots, X_m$ be independent and identically distributed (iid) random variables with the same
distribution as a non-negative random $X$ with mean $\mu=\mathbb{E}(X)>0$ and let  $X_{1:m}=\min\{X_1,\ldots,X_m\}\leqslant X_{2:m}\leqslant \cdots \leqslant X_{j:m} \leqslant \cdots \leqslant X_{k:m} \leqslant \cdots \leqslant X_{m-1:m} \leqslant X_{m:m}=\max\{X_1,\ldots,X_m\}$, for $1\leqslant j\leqslant k\leqslant m$, be the order statistics. For each integer $m\geqslant 2$, the {\bf extended Gini index} of $X$ is defined as follows:
\begin{align}\label{extended-Gini}
	IG_m(j,k)
	\equiv
	IG_m(j,k)(X)
	=
	\dfrac{\mathbb{E}[X_{k:m}-X_{j:m}]}{m\mu},
	\quad 
	1\leqslant j\leqslant k\leqslant m.
\end{align}

Summing over $k = 1, \ldots, m - 1$ after setting $j = k + 1$ in \eqref{extended-Gini}, the expression
\begin{align}\label{def-gini-index}
	\sum_{k=1}^{m-1}
	IG_m(k+1,k)
	=
	\dfrac{\mathbb{E}[X_{m:m}-X_{1:m}]}{m\mu}
	=
	IG_m(1,m)
	\equiv 
	IG_m
\end{align}
reduces to $m$th Gini index, $IG_m$, originally introduced by \cite{Gavilan-Ruiz2024}.

Consequently, by taking $m = 2=k$ and $j=1$ in \eqref{extended-Gini}, the extended Gini index  reduces to the standard Gini coefficient \citep{Gini1936}:
\begin{align*}
	G\equiv IG_2(1,2)= {\mathbb{E}\left[\vert X_1-X_2\vert\right]\over 2\mu}.
\end{align*}

\begin{proposition}\label{ext-gini-index-0}
For any $1\leqslant j\leqslant k\leqslant m$, the extended Gini index \eqref{extended-Gini} can be computed as
{\small
\begin{align*}
	IG_m(j,k)
	=
	\dfrac{		\displaystyle 
		\sum_{r=k}^{m}
		(-1)^{r-k}
		\binom{r-1}{k-1}
		\binom{m}{r}
		\int_0^\infty 
		\left[1-\{\mathbb{P}\left(X\leqslant t\right)\}^r\right]
		{\rm d}t
%		\\[0,2cm]
%		&
		-	
		\sum_{s=j}^{m}
		(-1)^{s-j}
		\binom{s-1}{j-1}
		\binom{m}{s}
		\int_0^\infty 
		\left[1-\{\mathbb{P}\left(X\leqslant t\right)\}^s\right]
		{\rm d}t}{
		\displaystyle
		m\int_0^\infty 
		\left[1-\mathbb{P}\left(X\leqslant t\right)\right]
		{\rm d}t}.
\end{align*}
}
\end{proposition}
\begin{proof}
Applying the results from \citet[][Proposition 10]{Salama-Koch2019}:
\begin{align*}
	X_{k:m}
	&=
	\sum_{r=k}^{m}
	(-1)^{r-k}
	\binom{r-1}{k-1}
	\sum_{1\leqslant t_1<\cdots<t_r\leqslant m}
	\max\{X_{t_1},\ldots,X_{t_r}\},
	\\[0,2cm]
		X_{j:m}
	&=
	\sum_{s=j}^{m}
	(-1)^{s-j}
	\binom{s-1}{j-1}
	\sum_{1\leqslant u_1<\cdots<u_s\leqslant m}
	\max\{X_{u_1},\ldots,X_{u_s}\},
\end{align*}
we have
\begin{align*}
	\mathbb{E}[X_{k:m}-X_{j:m}]
	&=
	\sum_{r=k}^{m}
(-1)^{r-k}
\binom{r-1}{k-1}
\sum_{1\leqslant t_1<\cdots<t_r\leqslant m}
\mathbb{E}[\max\{X_{t_1},\ldots,X_{t_r}\}]
\\[0,2cm]
&-	
	\sum_{s=j}^{m}
(-1)^{s-j}
\binom{s-1}{j-1}
\sum_{1\leqslant u_1<\cdots<u_s\leqslant m}
\mathbb{E}[\max\{X_{u_1},\ldots,X_{u_s}\}].
\end{align*}
Using the fact that $X_1, X_2, \ldots, X_m$ are iid, the above identity becomes
\begin{align*}
	\mathbb{E}[X_{k:m}-X_{j:m}]
	&=
	\sum_{r=k}^{m}
	(-1)^{r-k}
	\binom{r-1}{k-1}
	\binom{m}{r}
	\mathbb{E}[X_{r:r}]
	-	
	\sum_{s=j}^{m}
	(-1)^{s-j}
	\binom{s-1}{j-1}
	\binom{m}{s}
	\mathbb{E}[X_{s:s}].
\end{align*}
By applying the identity
\begin{align}\label{id-maximum}
	X_{r:r}
	=
	\max\{X_1,\ldots,X_r\}
	=
	\int_0^\infty 
	\left[1-\mathds{1}_{\bigcap_
		{i=1}^r\{X_i\leqslant t\}}\right]
	{\rm d}t,
\end{align}
we obtain
\begin{align}\label{exp-min}
	\mathbb{E}[X_{k:m}-X_{j:m}]
	&=
	\sum_{r=k}^{m}
	(-1)^{r-k}
	\binom{r-1}{k-1}
	\binom{m}{r}
	\int_0^\infty 
	\left[1-\mathbb{P}\left(\bigcap_
		{i=1}^r\{X_i\leqslant t\}\right)\right]
		{\rm d}t
		\nonumber
		\\[0,2cm]
	&-	
	\sum_{s=j}^{m}
	(-1)^{s-j}
	\binom{s-1}{j-1}
	\binom{m}{s}
	\int_0^\infty 
	\left[1-\mathbb{P}\left(\bigcap_
		{i=1}^s\{X_i\leqslant t\}\right)\right]
		{\rm d}t
		\nonumber
		\\[0,2cm]		
		&=
		\sum_{r=k}^{m}
		(-1)^{r-k}
		\binom{r-1}{k-1}
		\binom{m}{r}
		\int_0^\infty 
		\left[1-\{\mathbb{P}\left(X\leqslant t\right)\}^r\right]
		{\rm d}t
		\nonumber
		\\[0,2cm]
		&-	
		\sum_{s=j}^{m}
		(-1)^{s-j}
		\binom{s-1}{j-1}
		\binom{m}{s}
		\int_0^\infty 
		\left[1-\{\mathbb{P}\left(X\leqslant t\right)\}^s\right]
		{\rm d}t,
\end{align}
where Tonelli's Theorem validates the change in integration order, and the final step follows from the independence and identical distribution of $X_1, X_2, \ldots, X_m$.

Finally, combining the identities \eqref{exp-min}  and $\mu = \mathbb{E}(X) = \int_0^\infty [1 - \mathbb{P}(X \leqslant t)] {\rm d}t$ with the definition of the extended Gini index \eqref{extended-Gini}, the result follows.
\end{proof}

\begin{proposition}\label{ext-gini-index}
	For any $1\leqslant j\leqslant k\leqslant m$, the extended Gini index for $X\sim \text{Gamma}(\alpha,\lambda)$ (gamma distribution) is given by
	\begin{align}\label{mthgini_gamma}
		IG_m(j,k)
		&=
		{1\over\alpha m}
				\displaystyle 
			\sum_{r=k}^{m}
			(-1)^{r-k}
			\binom{r-1}{k-1}
			\binom{m}{r}
			\int_0^\infty 
			\left\{1-{\gamma^r(\alpha, t)\over\Gamma^r(\alpha) }\right\}
			{\rm d}t
				\nonumber	\\[0,2cm]
					&
			-	
					{1\over\alpha m}
			\sum_{s=j}^{m}
			(-1)^{s-j}
			\binom{s-1}{j-1}
			\binom{m}{s}
			\int_0^\infty 
			\left\{1-{\gamma^s(\alpha, t)\over\Gamma^s(\alpha) }\right\}
			{\rm d}t.
	\end{align}
\end{proposition}
\begin{proof}
The result follows directly from the Proposition \ref{ext-gini-index-0} and is thus omitted.
\end{proof}

\begin{remark}\label{rem-i}
Substituting $j=1$ and $k=m$ into Proposition \ref{ext-gini-index} and applying Newton's binomial Theorem yields
	\begin{align*}
IG_m
\equiv 
{IG}_m(1,m)
=
			{1\over \alpha m}
	\left[
	\int_0^\infty 
	\left\{
	1
	-
	{\gamma^m(\alpha,v)\over\Gamma^m(\alpha)}
	\right\}
	{\rm d}v  \,
	\nonumber
	-
	\int_0^\infty 
	\left\{
	1
	-
	{\gamma(\alpha,v)\over\Gamma(\alpha)}
	\right\}^m
	{\rm d}v
	\right],
		\end{align*}
	which is a previously established result in \cite{Vila2025}.
\end{remark}

\begin{remark}
	It is well-known that \cite[][Item 13]{Vila2025}
	\begin{align*}
		{	\displaystyle
			\int_0^\infty 
			\{
			\Gamma^2(\alpha)-\gamma^2(\alpha,s)
			\}
			{\rm d}s
			-
			\int_0^\infty 
			\Gamma^2(\alpha,s) {\rm d}s
		}
		=
		{\Gamma(2\alpha+1)\over 2^{2\alpha}\alpha^2\Gamma^2(\alpha)}.
	\end{align*}
	By applying the Legendre duplication formula  $\Gamma(x)\Gamma\left(x+{1}/{2}\right) = 2^{1-2x}\sqrt{\pi}\Gamma(2x)$ in the above identity, and then, by using Proposition \ref{ext-gini-index}, yields the following  widely recognized expression for the extended Gini index \citep{McDonald1979}:
	\begin{align}\label{main-formula}
		G
		\equiv 
		IG_2(1,2)
		=
		IG_2
		=
		{\Gamma(\alpha+{1\over 2})\over\sqrt{\pi}\alpha\Gamma(\alpha)}.
	\end{align}
\end{remark}

\section{Unbiasedness of extended Gini index estimator}\label{sec:03}

The main goal of this section is to derive a compact, closed-form expression for the expected value of a novel {\bf extended Gini index estimator},  $\widehat{IG}_m(j,k)$, $1\leqslant j\leqslant k\leqslant m$, which is defined as follows:
\begin{align}\label{estimator}
	\widehat{IG}_m(j,k)
	=
	{(m-1)!\over (n-1)(n-2)\cdots(n-m+1)} \,
	\dfrac{\displaystyle
		\sum_{\substack{{\bf i}=
				(i_1,\ldots,i_m)\in\mathbb{N}^m: \\[0,1cm]
			1\leqslant i_1<\cdots< i_m\leqslant n}}
		\left[
		X_{k:{\bf i}}
		-
		X_{j:{\bf i}}
		\right]
	}{\displaystyle \sum_{i=1}^{n}X_i},
\end{align}
where 
$
X_{1:{\bf i}}
=
\min\{X_{i_1},\ldots,X_{i_m}\}\leqslant X_{2:{\bf i}}\leqslant \cdots \leqslant X_{j:{\bf i}} \leqslant\cdots \leqslant X_{k:{\bf i}} \leqslant \cdots \leqslant X_{m-1:{\bf i}} \leqslant X_{m:{\bf i}}=\max\{X_{i_1},\ldots,X_{i_m}\}
$
are  the order statistics taken from
iid observations $X_{i_1}, X_{i_2},\ldots, X_{i_m}$ (of $X$). 

\begin{remark}\label{rem-gini-index}
	Setting $j=1$ and $k=m$ into \eqref{estimator} we obtain the $m$th Gini index estimator proposed in \cite{Vila2025}:
	\begin{align*}
			\widehat{IG}_m
			\equiv
	\widehat{IG}_m(1,m)
	=
		{(m-1)!\over (n-1)(n-2)\cdots(n-m+1)} \,
	\dfrac{\displaystyle
		\sum_{1\leqslant i_1<\cdots< i_m\leqslant n}
		\left[
		\max\{X_{i_1},\ldots,X_{i_m}\}
		-
		\min\{X_{i_1},\ldots,X_{i_m}\}
		\right]
	}{\displaystyle \sum_{i=1}^{n}X_i}.
	\end{align*}
\end{remark}

\begin{theorem}\label{main-theorem}
Let $X_1, X_2, \ldots, X_m$ be independent copies of a non-negative and absolutely continuous random variable $X$ with finite and positive expected value and common cumulative distribution function $F$. For any $1\leqslant j\leqslant k\leqslant m$, the following holds:
\begin{align*}
\mathbb{E}[	\widehat{IG}_m(j,k)]
	&=
	{n\over m}
	\sum_{r=k}^{m}
	(-1)^{r-k}
	\binom{r-1}{k-1}
	\binom{m}{r}
	\int_0^\infty
		\int_0^\infty 
\left\{
\mathscr{L}_F^{r}(z)
-
\mathbb{E}^r\left[
\mathds{1}_{\{X\leqslant t\}}
\exp\left(-X z\right)
\right]
\right\}
{\rm d}t  \,
	\mathscr{L}_F^{n-r}(z)
	{\rm d}z
	\nonumber
	\\[0,2cm]
	&-
	{n\over m}
	\sum_{s=j}^{m}
	(-1)^{s-j}
	\binom{s-1}{j-1}
	\binom{m}{s}
	\int_0^\infty 
		\int_0^\infty 
\left\{
\mathscr{L}_F^{s}(z)
-
\mathbb{E}^s\left[
\mathds{1}_{\{X\leqslant t\}}
\exp\left(-X z\right)
\right]
\right\}
{\rm d}t \,
	\mathscr{L}_F^{n-s}(z)
	{\rm d}z,
\end{align*}
where $\mathscr{L}_F(z)=\int_0^\infty \exp(-zx){\rm d}F(x)$ is the Laplace transform associated with distribution $F$, under the assumption that the relevant expectations and improper integrals converge.
\end{theorem}
\begin{proof}
By using the well-known identity
\begin{align*}
	\int_{0}^\infty \exp(-w z){\rm d}z={1\over w},
	\quad w>0,
\end{align*}
with $w=\sum_{i=1}^{n}X_i$, we get
\begin{multline}\label{exp-1}
	\mathbb{E}\left[
	\dfrac{\displaystyle
		\sum_{\substack{{\bf i}=
		(i_1,\ldots,i_m)\in\mathbb{N}^m: \\[0,1cm]
		1\leqslant i_1<\cdots< i_m\leqslant n}}
\left(
X_{k:{\bf i}}
-
X_{j:{\bf i}}
\right)
	}{\displaystyle \sum_{i=1}^{n}X_i}
	\right]
	=
	\sum_{\substack{{\bf i}=
		(i_1,\ldots,i_m)\in\mathbb{N}^m: \\[0,1cm]
		1\leqslant i_1<\cdots< i_m\leqslant n}}
	\mathbb{E}\left[X_{k:{\bf i}}\int_0^\infty\exp\left\{-\left(\sum_{i=1}^{n}X_i\right)z\right\}{\rm d}z\right]
	\\[0,2cm]
	-
	\sum_{\substack{{\bf i}=
		(i_1,\ldots,i_m)\in\mathbb{N}^m: \\[0,1cm]
		1\leqslant i_1<\cdots< i_m\leqslant n}}
	\mathbb{E}\left[X_{j:{\bf i}}\int_0^\infty\exp\left\{-\left(\sum_{i=1}^{n}X_i\right)z\right\}{\rm d}z\right]
		\\[0,2cm]
		=
	\sum_{\substack{{\bf i}=
		(i_1,\ldots,i_m)\in\mathbb{N}^m: \\[0,1cm]
		1\leqslant i_1<\cdots< i_m\leqslant n}} \
\int_0^\infty
\mathbb{E}\left[X_{k:{\bf i}}
\exp\left\{-\left(\sum_{i=1}^{n}X_i\right)z\right\}
\right]
{\rm d}z
\\[0,2cm]
-
\sum_{\substack{{\bf i}=
		(i_1,\ldots,i_m)\in\mathbb{N}^m: \\[0,1cm]
		1\leqslant i_1<\cdots< i_m\leqslant n}} \
\int_0^\infty
\mathbb{E}\left[X_{j:{\bf i}}
\exp\left\{-\left(\sum_{i=1}^{n}X_i\right)z\right\}
\right]
{\rm d}z,
\end{multline}
where Tonelli's Theorem justifies the interchange of integrals.
Utilizing the established identities from \citet[][Proposition 10]{Salama-Koch2019}:
\begin{align*}
	X_{k:{\bf i}}
	&=
		\sum_{r=k}^{m}
	(-1)^{r-k}
	\binom{r-1}{k-1}
	\sum_{\substack{{\bf t}=(t_1,\ldots,t_r)\in\mathbb{N}^r: \\[0,1cm]
			1\leqslant t_1<\cdots<t_r\leqslant m}}
	\max\{X_{t_1},\ldots,X_{t_r}\}
	\\[0,2cm]
	&=
	\sum_{r=k}^{m}
	(-1)^{r-k}
	\binom{r-1}{k-1}
	\sum_{\substack{{\bf t}=(t_1,\ldots,t_r)\in\mathbb{N}^r: \\[0,1cm]
		1\leqslant t_1<\cdots<t_r\leqslant m}}
	X_{r:{\bf t}}
\end{align*}
and
\begin{align*}
	X_{j:{\bf i}}
&=
\sum_{s=j}^{m}
(-1)^{s-j}
\binom{s-1}{j-1}
\sum_{\substack{{\bf u}=(u_1,\ldots,u_s)\in\mathbb{N}^s: \\[0,1cm]
		1\leqslant u_1<\cdots<u_s\leqslant m}}
		\max\{X_{u_1},\ldots,X_{u_s}\}
	\\[0,2cm]
&=
\sum_{s=j}^{m}
(-1)^{s-j}
\binom{s-1}{j-1}
\sum_{\substack{{\bf u}=(u_1,\ldots,u_s)\in\mathbb{N}^s: \\[0,1cm]
		1\leqslant u_1<\cdots<u_s\leqslant m}}
X_{s:{\bf u}},
\end{align*}
the expression \eqref{exp-1} becomes
\begin{multline}\label{exp-2}
		=
\sum_{\substack{{\bf i}=
		(i_1,\ldots,i_m)\in\mathbb{N}^m: \\[0,1cm]
		1\leqslant i_1<\cdots< i_m\leqslant n}}
	\sum_{r=k}^{m}
(-1)^{r-k}
\binom{r-1}{k-1}
\sum_{\substack{{\bf t}=(t_1,\ldots,t_r)\in\mathbb{N}^r: \\[0,1cm]
		1\leqslant t_1<\cdots<t_r\leqslant m}} \
\int_0^\infty
\mathbb{E}\left[
X_{m:{\bf t}}
\exp\left\{-\left(\sum_{i=1}^{r}X_i\right)z\right\}
\exp\left\{-\left(\sum_{i=r+1}^{n}X_i\right)z\right\}
\right]
{\rm d}z
\\[0,2cm]
-
\sum_{\substack{{\bf i}=
		(i_1,\ldots,i_m)\in\mathbb{N}^m: \\[0,1cm]
		1\leqslant i_1<\cdots< i_m\leqslant n}}
\sum_{s=j}^{m}
(-1)^{s-j}
\binom{s-1}{j-1}
\sum_{\substack{{\bf u}=(u_1,\ldots,u_s)\in\mathbb{N}^s: \\[0,1cm]
		1\leqslant u_1<\cdots<u_s\leqslant m}} \
\int_0^\infty
\mathbb{E}\left[
X_{m:{\bf u}}
\exp\left\{-\left(\sum_{i=1}^{s}X_i\right)z\right\}
\exp\left\{-\left(\sum_{i=s+1}^{n}X_i\right)z\right\}
\right]
{\rm d}z.
\end{multline}
Given the identical distribution of $X_1, X_2, \ldots, X_m$, we have
\begin{align*}
X_{m:{\bf t}}
\exp\left\{-\left(\sum_{i=1}^{r}X_i\right)z\right\}
&=
\max\{X_{t_1},\ldots,X_{t_r}\}
\exp\left\{-\left(\sum_{i=1}^{r}X_i\right)z\right\}
	\\[0,2cm]
&
\stackrel{d}{=}
\max\{X_{1},\ldots,X_{r}\}
\exp\left\{-\left(\sum_{i=1}^{r}X_i\right)z\right\}
=
X_{r:r}
\exp\left\{-\left(\sum_{i=1}^{r}X_i\right)z\right\}
\end{align*}
and
\begin{align*}
X_{m:{\bf u}}
\exp\left\{-\left(\sum_{i=1}^{s}X_i\right)z\right\}
&=
\max\{X_{u_1},\ldots,X_{u_s}\}
\exp\left\{-\left(\sum_{i=1}^{s}X_i\right)z\right\}
	\\[0,2cm]
&
\stackrel{d}{=}
\max\{X_{1},\ldots,X_{s}\}
\exp\left\{-\left(\sum_{i=1}^{s}X_i\right)z\right\}
=
X_{s:s}
\exp\left\{-\left(\sum_{i=1}^{s}X_i\right)z\right\},
\end{align*}
where $\stackrel{d}{=}$ denotes equality in distribution of random variables. Applying the aforementioned identities to \eqref{exp-2} and leveraging independence, the expression simplifies to
\begin{align*}
	&=	
	\binom{n}{m}
	\sum_{r=k}^{m}
	(-1)^{r-k}
	\binom{r-1}{k-1}
	\binom{m}{r}
	\int_0^\infty
	\mathbb{E}\left[
	X_{r:r}
	\exp\left\{-\left(\sum_{i=1}^{r}X_i\right)z\right\}
	\mathbb{E}\left[
	\exp\left\{-\left(\sum_{i=r+1}^{n}X_i\right)z\right\}
	\right]
	\right]
	{\rm d}z
	\\[0,2cm]
	&-
	\binom{n}{m}
	\sum_{s=j}^{m}
	(-1)^{s-j}
	\binom{s-1}{j-1}
	\binom{m}{s}
	\int_0^\infty
	\mathbb{E}\left[
	X_{s:s}
	\exp\left\{-\left(\sum_{i=1}^{s}X_i\right)z\right\}
	\mathbb{E}
	\left[
	\exp\left\{-\left(\sum_{i=s+1}^{n}X_i\right)z\right\}
	\right]
	\right]
	{\rm d}z
		\\[0,2cm]
	&=
\binom{n}{m}
\sum_{r=k}^{m}
(-1)^{r-k}
\binom{r-1}{k-1}
\binom{m}{r}
\int_0^\infty
\mathbb{E}\left[
X_{r:r}
\exp\left\{-\left(\sum_{i=1}^{r}X_i\right)z\right\}
\right]
\mathscr{L}_F^{n-r}(z)
{\rm d}z
\\[0,2cm]
&-
\binom{n}{m}
\sum_{s=j}^{m}
(-1)^{s-j}
\binom{s-1}{j-1}
\binom{m}{s}
\int_0^\infty
\mathbb{E}\left[
X_{s:s}
\exp\left\{-\left(\sum_{i=1}^{s}X_i\right)z\right\}
\right]
\mathscr{L}_F^{n-s}(z)
{\rm d}z.
\end{align*}

Thus, we have demonstrated that
\begin{align}\label{exp-3}
	\mathbb{E}\left[
	\dfrac{\displaystyle
		\sum_{\substack{{\bf i}=
				(i_1,\ldots,i_m)\in\mathbb{N}^m: \\[0,1cm]
				1\leqslant i_1<\cdots< i_m\leqslant n}}
		\left(
		X_{k:{\bf i}}
		-
		X_{j:{\bf i}}
		\right)
	}{\displaystyle \sum_{i=1}^{n}X_i}
	\right]
		&=
	\binom{n}{m}
	\sum_{r=k}^{m}
	(-1)^{r-k}
	\binom{r-1}{k-1}
	\binom{m}{r}
	\int_0^\infty
	\mathbb{E}\left[
	X_{r:r}
	\exp\left\{-\left(\sum_{i=1}^{r}X_i\right)z\right\}
	\right]
	\mathscr{L}_F^{n-r}(z)
	{\rm d}z
	\nonumber
	\\[0,2cm]
	&-
	\binom{n}{m}
	\sum_{s=j}^{m}
	(-1)^{s-j}
	\binom{s-1}{j-1}
	\binom{m}{s}
	\int_0^\infty
	\mathbb{E}\left[
	X_{s:s}
	\exp\left\{-\left(\sum_{i=1}^{s}X_i\right)z\right\}
	\right]
	\mathscr{L}_F^{n-s}(z)
	{\rm d}z.
\end{align}
	
On the other hand, by using the identity in \eqref{id-maximum},
%\begin{align*}
%	X_{r:r}
%	=
%	\max\{X_1,\ldots,X_r\}
%	=
%	\int_0^\infty 
%	\left[1-\mathds{1}_{\bigcap_
%		{i=1}^r\{X_i\leqslant t\}}\right]
%	{\rm d}t,
%\end{align*}
we have
\begin{align}\label{exp-11}
	\mathbb{E}\left[
	X_{r:r}
	\exp\left\{-\left(\sum_{i=1}^{r}X_i\right)z\right\}
	\right]
	&=
	\int_0^\infty 
	\mathbb{E}\left[
	\left(1-\mathds{1}_{\bigcap_
		{i=1}^r\{X_i\leqslant t\}}\right)
	\exp\left\{-\left(\sum_{i=1}^{r}X_i\right)z\right\}
	\right]
	{\rm d}t
	\nonumber
	\\[0,2cm]
	&=
	\int_0^\infty 
	\left\{
	\mathscr{L}_F^{r}(z)
	-
	\mathbb{E}\left[
	\mathds{1}_{\bigcap_
		{i=1}^r\{X_i\leqslant t\}}
	\exp\left\{-\left(\sum_{i=1}^{r}X_i\right)z\right\}
	\right]
	\right\}
	{\rm d}t
		\nonumber
	\\[0,2cm]
	&=
		\int_0^\infty 
	\left\{
	\mathscr{L}_F^{r}(z)
	-
	\mathbb{E}^r\left[
	\mathds{1}_{\{X\leqslant t\}}
	\exp\left(-X z\right)
	\right]
	\right\}
	{\rm d}t,
\end{align}
where the last equality follows from the fact that the variables $X_1, X_2, \ldots, X_m$ are identically distributed as $X$.

Analogously, we get
\begin{align}\label{exp-12}
	\mathbb{E}\left[
X_{s:s}
\exp\left\{-\left(\sum_{i=1}^{s}X_i\right)z\right\}
\right]
	=
		\int_0^\infty 
\left\{
\mathscr{L}_F^{s}(z)
-
\mathbb{E}^s\left[
\mathds{1}_{\{X\leqslant t\}}
\exp\left(-X z\right)
\right]
\right\}
{\rm d}t.
\end{align}

Plugging \eqref{exp-11} and \eqref{exp-12} into \eqref{exp-3},
% and then applying Tonelli's Theorem to justify interchanging the integrals, 
the proof of the theorem follows.
\end{proof}

We now apply Theorem \ref{main-theorem} to obtain an explicit formula for the expectation of the estimator $\widehat{IG}_m(j,k)$ in gamma populations, thereby establishing its unbiasedness.

\begin{corollary}\label{corollary-main}
	Let $X_1, X_2,\ldots, X_m$ be independent copies of $X\sim\text{Gamma}(\alpha,\lambda)$. For any $1\leqslant j\leqslant k\leqslant m$, we have:
	\begin{align*}
\mathbb{E}[\widehat{IG}_m(j,k)]
&=
{1\over \alpha m}
\sum_{r=k}^{m}
(-1)^{r-k}
\binom{r-1}{k-1}
\binom{m}{r}
\int_0^\infty 
\left\{
1
-
{\gamma^r(\alpha,v)\over\Gamma^r(\alpha)}
\right\}
{\rm d}v  \,
\nonumber
\\[0,2cm]
&-
{1\over \alpha m}
\sum_{s=j}^{m}
(-1)^{s-j}
\binom{s-1}{j-1}
\binom{m}{s}
\int_0^\infty 
\left\{
1
-
{\gamma^s(\alpha,v)\over\Gamma^s(\alpha)}
\right\}
{\rm d}v
=
{IG}_m(j,k),
\end{align*}
where ${IG}_m(j,k)$ is the extended Gini index given in Proposition \ref{ext-gini-index}. Thus, the estimator $	\widehat{IG}_m(j,k)$ is unbiased
for gamma populations.	
\end{corollary}
\begin{proof}
For $X\sim\text{Gamma}(\alpha,\lambda)$ direct computation gives
\begin{align*}
\mathbb{E}\left[
\mathds{1}_{\{X\leqslant t\}}
\exp\left(-X z\right)
\right]
=
{\lambda^\alpha\over (z+\lambda)^\alpha\Gamma(\alpha)} \, \gamma(\alpha,(z+\lambda)t).
\end{align*}
As $\mathscr{L}_F(z)=\lambda^\alpha/(z+\lambda)^\alpha$, Theorem \ref{main-theorem} yields
\begin{align*}
	\mathbb{E}[\widehat{IG}_m(j,k)]
	&=
	{n\lambda^{\alpha n}\over m}
	\sum_{r=k}^{m}
	(-1)^{r-k}
	\binom{r-1}{k-1}
	\binom{m}{r}
	\int_0^\infty
	\int_0^\infty 
	\left\{
	1
	-
{\gamma^r(\alpha,(z+\lambda)t)\over\Gamma^r(\alpha)}
	\right\}
	{\rm d}t  \,
	{1\over (z+\lambda)^{\alpha n}} \,
	{\rm d}z
	\nonumber
	\\[0,2cm]
	&-
	{n\lambda^{\alpha n}\over m}
	\sum_{s=j}^{m}
	(-1)^{s-j}
	\binom{s-1}{j-1}
	\binom{m}{s}
	\int_0^\infty 
	\int_0^\infty 
	\left\{
	1
-
{\gamma^s(\alpha,(z+\lambda)t)\over\Gamma^s(\alpha)}
	\right\}
	{\rm d}t \,
	{1\over (z+\lambda)^{\alpha n}} \,
	{\rm d}z.
\end{align*}
With the change of variable $v= (z + \lambda)t$, the above identity becomes
\begin{align*}
\mathbb{E}[\widehat{IG}_m(j,k)]
&=
{n\lambda^{\alpha n}\over m}
\sum_{r=k}^{m}
(-1)^{r-k}
\binom{r-1}{k-1}
\binom{m}{r}
\int_0^\infty
\int_0^\infty 
\left\{
1
-
{\gamma^r(\alpha,v)\over\Gamma^r(\alpha)}
\right\}
{\rm d}v  \,
{1\over (z+\lambda)^{\alpha n+1}} \,
{\rm d}z
\nonumber
\\[0,2cm]
&-
{n\lambda^{\alpha n}\over m}
\sum_{s=j}^{m}
(-1)^{s-j}
\binom{s-1}{j-1}
\binom{m}{s}
\int_0^\infty 
\int_0^\infty 
\left\{
1
-
{\gamma^s(\alpha,v)\over\Gamma^s(\alpha)}
\right\}
{\rm d}v \,
{1\over (z+\lambda)^{\alpha n+1}} \,
{\rm d}z
\\[0,2cm]
&=
{1\over \alpha m}
\sum_{r=k}^{m}
(-1)^{r-k}
\binom{r-1}{k-1}
\binom{m}{r}
\int_0^\infty 
\left\{
1
-
{\gamma^r(\alpha,v)\over\Gamma^r(\alpha)}
\right\}
{\rm d}v  \,
\nonumber
\\[0,2cm]
&-
{1\over \alpha m}
\sum_{s=j}^{m}
(-1)^{s-j}
\binom{s-1}{j-1}
\binom{m}{s}
\int_0^\infty 
\left\{
1
-
{\gamma^s(\alpha,v)\over\Gamma^s(\alpha)}
\right\}
{\rm d}v.
\end{align*}
Hence, from Proposition \ref{ext-gini-index} the proof is concluded.
\end{proof}

\begin{remark}
	Note that Corollary \ref{corollary-main}'s result generalizes prior findings by \cite{Deltas2003}, \cite{Baydil2025}, and \cite{Vila2025}.
\end{remark}

\begin{remark}
	Due to the scale invariance of the extended Gini index estimator $\widehat{IG}_m(j,k)$, its expectation $\mathbb{E}[\widehat{IG}_m(j,k)]$ is independent of the rate parameter $\lambda$, as asserted in Corollary \ref{corollary-main}.
\end{remark}

Setting $j=1$ and $k=m$ into Corollary \ref{corollary-main}, from Newton’s binomial Theorem it follows that
\begin{proposition}\label{prop-main}
	Let $X_1, X_2,\ldots, X_m$ be independent copies of $X\sim\text{Gamma}(\alpha,\lambda)$. We have:
		\begin{align*}
		\mathbb{E}[\widehat{IG}_m]
		\equiv
		\mathbb{E}[\widehat{IG}_m(1,m)]
		&=
		{1\over \alpha m}
		\left[
		\int_0^\infty 
		\left\{
		1
		-
		{\gamma^m(\alpha,v)\over\Gamma^m(\alpha)}
		\right\}
		{\rm d}v  \,
		\nonumber
		-
		\int_0^\infty 
		\left\{
		1
		-
		{\gamma(\alpha,v)\over\Gamma(\alpha)}
		\right\}^m
		{\rm d}v
		\right]
		\\[0,2cm]
		&=
		{IG}_m(1,m)
		\equiv 
		IG_m,
	\end{align*}
	where $IG_m$ is the $m$th Gini index given in Remark \ref{rem-i} and $\widehat{IG}_m$ is its  corresponding estimator (initially introduced by \cite{Vila2025}) given in  Remark \ref{rem-gini-index}.
\end{proposition}

\begin{remark}
	Note that the result of Proposition \ref{prop-main} was previously established in \cite{Vila2025}.
\end{remark}

\section{Illustrative simulation study}\label{sec:04}

In this section, we provide a Monte Carlo simulation study to evaluate the finite-sample performance of the extended Gini index estimator. Specifically, we want to confirm via simulations the unbiasedness of the extended Gini index estimator. We consider independent random samples of size $n \in \{5, 10, 20, 30\}$ drawn from a gamma distribution with shape parameter $\alpha = 2$ and rate parameter $\beta = 1$. For each sample, we compute the extended Gini index using $m = 5$, $j = 2$, and $k = 4$. The simulation is repeated $N_{\text{sim}} = 500$ times to estimate the bias and mean squared error (MSE) of the estimator, which are given by
\begin{eqnarray*}
 \widehat{\text{Bias}}(\widehat{IG}_m(j,k)) = \frac{1}{N_{\text{sim}}} \sum_{l=1}^{N_{\text{sim}}} {\widehat{IG}_m(j,k)^{(l)} - {IG}_m(j,k)},\\
\quad
\widehat{\text{MSE}}(\widehat{IG}_m(j,k)) = \frac{1}{N_{\text{sim}}} \sum_{l=1}^{N_{\text{sim}}} \big(\widehat{IG}_m(j,k)^{(l)} - {IG}_m(j,k)\big)^2,
\end{eqnarray*}
where $\widehat{IG}_m(j,k)^{(l)}$ denotes the $l$-th Monte Carlo replicate of the extended Gini index estimator defined in Equation~\eqref{estimator}, and ${IG}_m(j,k)$ represents the true extended Gini index defined in Equation~\eqref{mthgini_gamma}. Here, $N_{\text{sim}}$ is the number of Monte Carlo replications.

Table~\ref{tab:sim_results} summarizes the simulation results, including the sample size ($n$), bias, MSE, average estimates of the index, and its true value under the specified gamma distribution. From this table, we note that, as expected, the Monte Carlo results are consistent with the theoretical unbiasedness established in Section~\ref{sec:03}, i.e., the empirical bias is close to zero. We also note that the MSE is close to zero and decreases as the sample size grows.

\begin{table}[!ht]
\centering
\caption{Monte Carlo results for the extended Gini index estimator with $m = 4$, $j = 2$, $k = 3$, $\alpha = 2$, and $\beta = 1$, based on $500$ replications.}
\label{tab:sim_results}
\begin{tabular}{ccccc}
\toprule
$n$ & Bias & MSE & Average Estimate & True Value \\
\midrule
5  & $-3.87 \times 10^{-4}$ & $2.61 \times 10^{-3}$ & 0.09618 & 0.09657 \\
10 & $-4.22 \times 10^{-4}$ & $6.59 \times 10^{-4}$ & 0.09615 & 0.09657 \\
20 & $-5.44 \times 10^{-5}$ & $3.13 \times 10^{-4}$ & 0.09652 & 0.09657 \\
30 & $-1.12 \times 10^{-4}$ & $1.72 \times 10^{-4}$ & 0.09646 & 0.09657 \\
\bottomrule
\end{tabular}
\end{table}

\section{Application to real data}
\label{sec:05}

We illustrate the behavior of the extended  Gini index using real-world income data. Specifically, we consider data on gross domestic product (GDP) per capita (expressed in international-\$ at 2021 prices) for the year 2023; see \url{https://ourworldindata.org/grapher/gdp-per-capita-worldbank} and Table \ref{table_countries}. The analysis focuses on countries in South America, supplemented by three high-income countries (United States, Ireland, and Singapore) and three low-income countries (Somalia, Liberia, and Uganda), in order to create a broader context of income dispersion.

% latex table generated in R 4.5.0 by xtable 1.8-4 package
% Fri May  2 16:14:49 2025
\begin{table}[!ht]\label{table_countries}
\centering
\caption{2023 GDP per capita for selected countries.}
\begin{tabular}{lr}
  \hline
Countries ($n=17$) & GDP (international-\$ in 2021 prices) \\
  \hline
Singapore & 127543.55 \\
  Ireland & 114922.39 \\
  United States & 74577.51 \\
  Guyana & 49315.16 \\
  Uruguay & 31019.31 \\
  Chile & 29462.64 \\
  Argentina & 27104.98 \\
  Suriname & 19043.71 \\
  Brazil & 19018.24 \\
  Colombia & 18692.38 \\
  Paraguay & 15783.11 \\
  Peru & 15294.26 \\
  Ecuador & 14472.32 \\
  Bolivia & 9843.97 \\
  Uganda & 2791.06 \\
  Liberia & 1616.92 \\
  Somalia & 1402.47 \\
   \hline
\end{tabular}
\end{table}

We fitted a gamma distribution to the income data in Table~\ref{table_countries} using maximum likelihood estimation, implemented via the \texttt{R} package \texttt{fitdistrplus} \citep{fitdistrplus:15}. Figure~\ref{fig:diagnostic_plot} displays the corresponding diagnostic plots, which suggest a good fit between the empirical data and the theoretical distribution. To further evaluate the adequacy of the gamma model, we performed two goodness-of-fit tests: the Kolmogorov–Smirnov (KS) and the Cramér–von Mises (CvM) tests. The resulting $p$-values are 0.7465 and 0.7348, respectively, indicating no evidence against the gamma distribution hypothesis. Therefore, these results support the adequacy of the gamma distribution assumption.

\begin{figure}[H]
\centering
\includegraphics[width=0.8\textwidth]{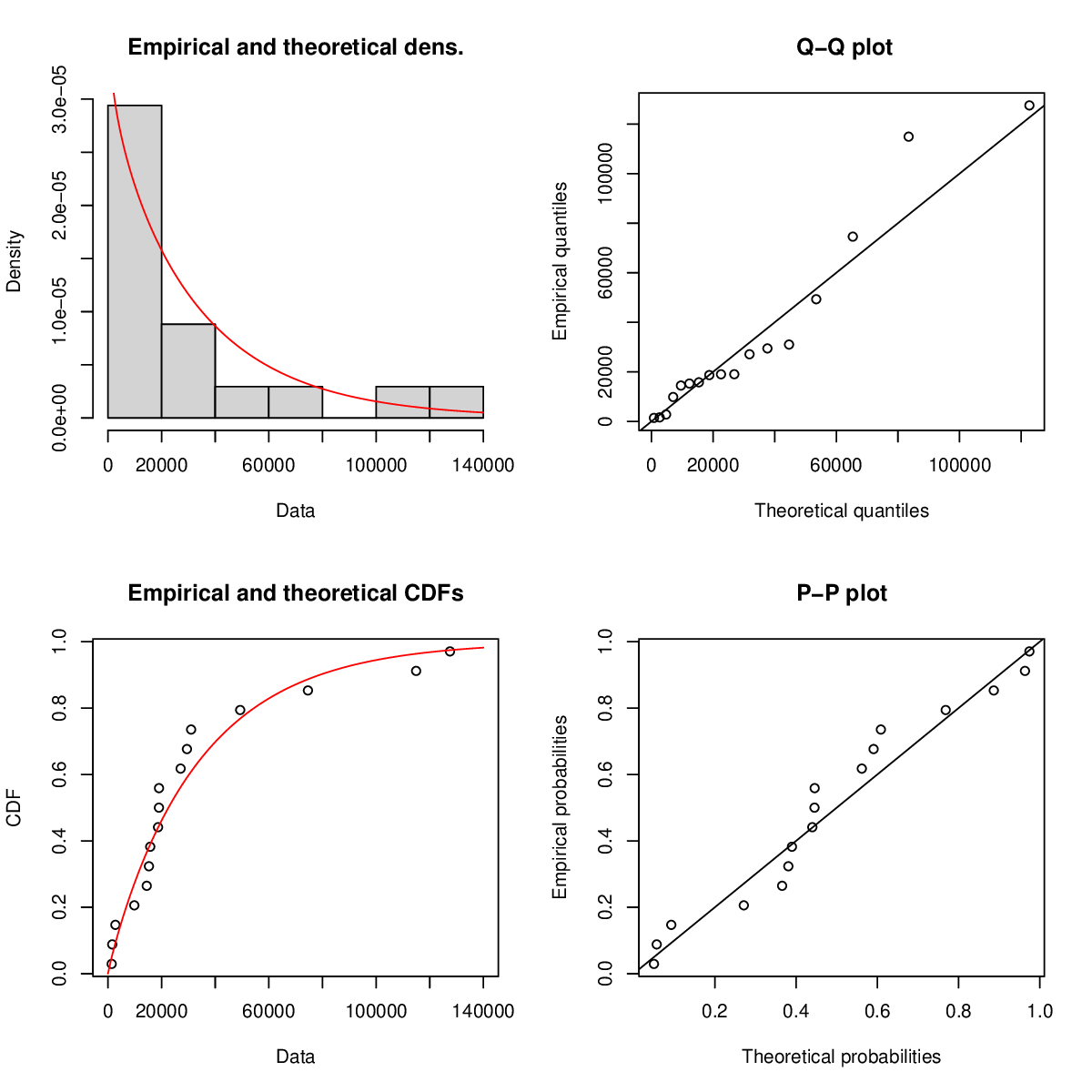}
\caption{Diagnostic plots for the gamma distribution fitted to 2023 GDP per capita data.}
\label{fig:diagnostic_plot}
\end{figure}

Figure~\ref{fig:heatmap_gmjk} shows a heatmap summarizing the estimated values of $G_m(j,k)$ (defined in Equation \eqref{estimator}) across multiple values of $m$, $j$ and $k$, satisfying $1 \leqslant j \leqslant k \leqslant m$. As mentioned, from Equation \eqref{estimator}, we obtain the, standard Gini index estimator when $m= 2=k$ and $j=1$ \citep{Gini1936}, and the  $m$th Gini index estimator when $j=1$ and $k=17$ \citep{Vila2025}. On the one hand, the left-most plot in Figure~\ref{fig:heatmap_gmjk} corresponds to the standard Gini index, estimated at 0.5600, which reflects a high level of income disparity across the sample. This is expected due to the inclusion of both high-income economies (United States, Ireland, and Singapore) and low-income countries (Somalia, Liberia, and Uganda). On the other hand, the value of the $m$th Gini index estimate is 0.2206, as shown in the bottom-right plot of the heatmap. These results emphasizes the effect of averaging differences over all possible subsamples of size $m$, thereby reducing the influence of extreme values. Our proposed index $G_m(j,k)$, as can be seen in the estimates plotted in Figure~\ref{fig:heatmap_gmjk}, allows a wide range of combinations of $m$, $j$ and $k$. Note that combinations that emphasize differences between the outermost order statistics tend to yield higher inequality estimates, while those focusing on central positions indicate lower inequality. This flexibility allows the analyst to control the influence of outliers and to extract more information on inequality concentration.

\begin{figure}[H]
\centering
\includegraphics[width=1.0\textwidth]{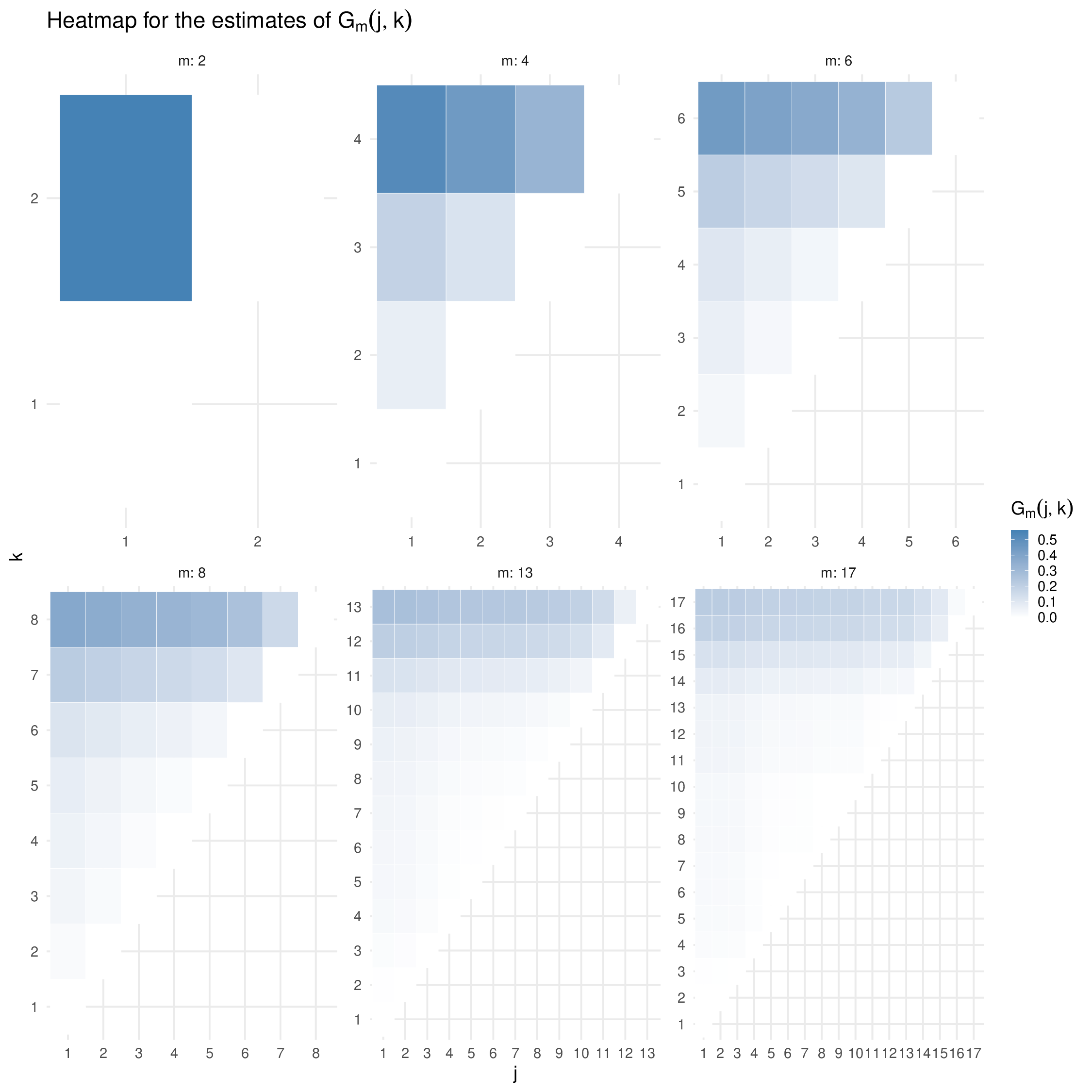}
\caption{Heatmap of the estimates of $G_m(j,k)$ for different values of $m$, based on 2023 GDP per capita data.}
\label{fig:heatmap_gmjk}
\end{figure}

\section{Concluding remarks}\label{sec:06}

We proposed a novel generalization of the $m$th Gini index \citep{Gavilan-Ruiz2024} based on ordered differences within subsets of size $m$, leading to the extended Gini index. This flexible formulation captures distributional features that the $m$th Gini index may overlook, particularly by allowing for the exclusion of extreme values through  choices of $j$ and $k$. We derived a closed-form unbiased estimator of this index for gamma-distributed populations and investigated its finite-sample properties via Monte Carlo simulations. The simulation confirmed the unbiasedness of the proposed estimator in finite samples. We applied the proposed estimator to real GDP per capita data for the year 2023, covering countries in South America, along with high- and low-income economies. The results showed that the extended Gini index provides a rich spectrum of inequality estimates. Particularly, combinations of $(j,k)$ that encompasses extremes higher inequality estimates, whereas those focusing on central order statistics return lower inequality values. This flexibility makes the proposed index a powerful tool for capturing distributional nuances frequently neglected by classical measures such as the classical Gini index.

%\clearpage
%\clearpage

\paragraph*{Acknowledgements}
The research was supported in part by CNPq and CAPES grants from the Brazilian government.

\paragraph*{Disclosure statement}
There are no conflicts of interest to disclose.

%%%%%%%%%%%%%%%%%%%%%%%%%%%%%%%%%%%%%%%%%%%%%%%%%%%%%%%%%%%%%

%%\bibliographystyle{apalike}
%%\bibliography{references}

\end{document}